\documentclass[accepted]{uai2023} 

\usepackage[british]{babel}

\usepackage{natbib} 
\usepackage{mathtools} 
\usepackage{booktabs} 
\usepackage{tikz} 
\usepackage{amsmath}
\usepackage{amsthm}
\usepackage{amssymb}
\usepackage{algpseudocode}
\usepackage{algorithm}
\usepackage{subcaption}

\newcommand{\E}{\mathbf{E}}
\newcommand{\EXP}{\mathrm{Exp}}

\newcommand{\Mlay}{{\mathrm{LAY}}}
\newcommand{\Mrec}{{\mathrm{REC}}}

\newcommand{\EMD}{{\mathrm{EMD}}}
\renewcommand{\Pr}{\mathrm{Pr}}
\newcommand{\R}{\mathbb{R}}

\newcommand{\BibTeX}{\rm B\kern-.05em{\sc i\kern-.025em b}\kern-.08em\TeX}
\newtheorem{theorem}{Theorem}[section]
\newtheorem{definition}[theorem]{Definition}
\newtheorem{lemma}[theorem]{Lemma}

\newtheorem{example}[theorem]{Example}
\newtheorem{corollary}[theorem]{Corollary}

\title{Differentially Private Diffusion Auction: The Single-unit Case}

\author[1]{Fengjuan Jia}
\author[1]{Mengxiao Zhang}
\author[2]{Jiamou Liu}
\author[1]{Bakh Khoussainov}
\affil[1]{School of Computer Science and Engineering, University of Electronic Science and Technology of China}
\affil[2]{School of Computer Science, The University of Auckland}

\begin{document}
\maketitle

\begin{abstract}
Diffusion auction refers to an emerging paradigm of online marketplace where an auctioneer utilises a social network to attract potential buyers.  Diffusion auction poses significant privacy risks. From the auction outcome, it is possible to infer hidden, and potentially sensitive, preferences of buyers. To mitigate such risks, we initiate the study of differential privacy (DP) in diffusion auction mechanisms. DP is a well-established notion of privacy that protects a system against inference attacks. Achieving DP in diffusion auctions is non-trivial as the well-designed auction rules are required to incentivise the buyers to truthfully report their neighbourhood. We study the single-unit case and design two differentially private diffusion mechanisms (DPDMs): recursive DPDM and layered DPDM. We prove that these mechanisms guarantee differential privacy, incentive compatibility and individual rationality for both valuations and neighbourhood. We then empirically compare their performance on real and synthetic datasets.
\end{abstract}

\section{Introduction}


New technological shift in AI and data science has given rise to an imminent need to address data privacy issues in online platforms. Indeed, a Gartner survey shows that $41\%$ of the surveyed organisations have experienced a privacy breach or security incident\footnote{https://blogs.gartner.com/avivah-litan/2022/08/05/ai-models-under-attack-conventional-controls-are-not-enough/}. Data privacy issues have been especially serious and impactful around the use of social commerce platforms such as Instagram and Facebook. As users of such a platform find, browse and buy products through the social network, they are also exposed to a significant risk of privacy leakage. 
A recent PCI Pal survey shows that fewer than $7\%$ of users are confident about their data security on social commerce sites\footnote{https://www.pcipal.com/knowledge-centre/resource/fewer-than-10-of-people-are-confident-about-their-data-security-on-social-media-according-to-survey-from-pci-pal/}. Thus designing new tools to facilitate safe and private use of social commerce platforms is of crucial importance.  


Auction is important  in facilitating online commerce.  Auctions have been applied in many contexts,  e.g.,  radio spectrum, sponsored search ads, virtual  resource allocation. In an auction, buyers  submit their (private) valuations in bids to the auctioneer.  The bids often
imply buyers' preferences and confidential business strategies, and competitors may exploit them to gain an advantage. Hence, there is a need to protect the privacy of bid information. 
%
The privacy issues in auctions have recently been studied in \citep{mcsherry2007mechanism,Lin2018Frameworks,zhu2014differentially,ni2021differentially}. To mitigate privacy risks, these studies employ the well-established notion of {\em differential privacy} (DP) \citep{dwork2006calibrating}. Here, DP is used to protect individual's bid information when the auction outcome is published.  
To achieve DP on bids, the work of \cite{mcsherry2007mechanism} proposed {\em exponential mechanism}.  The mechanism randomises auction results so that a change in a buyer's bid does not significantly affect the auction outcome. In this way, the mechanism prevents the bid from being inferred from the auction outcome. 
This mechanism has so far been a predominant method to protect privacy in auctions.


{\em Diffusion auction} is an emerging form of auction. 
In this setting, a seller is able to harness the power of social network to diffuse auction information, inviting friends, friends-of-friends, etc., to join the auction, thereby attracting a large number of potential buyers. This differs from a standard auction (without social network) where the participants are fixed beforehand. Thus, diffusion auction are especially suitable for facilitating online social commerce platforms where the social network plays a prominent role. 
A challenge in diffusion auctions lies in resolving the conflict between the seller who wants to attract more participants for better revenue and the buyers who are reluctant to invite their friends to avoid competition. Thus there is a need to extend {\em incentive compatibility} (IC) for hidden valuations in classical auctions, to {\em diffusion IC} for hidden valuation as well as social ties. Numerous studies, e.g., \citep{li2017mechanism,li2019diffusion,zhang2020incentivize,zhang2020redistribution}, have proposed mechanisms for diffusion auction that achieve diffusion IC. 

Diffusion  auctions are prone to all aforementioned privacy risks for auctions in general.  
However, no study has focused on the privacy issues for diffusion auctions. 
Here we close this gap by investigating the following question: 
\vspace{-0.5cm}
\begin{quote}
    {\em How do we design a differentially private diffusion mechanism (DPDM) that guarantees desirable properties and preserves valuation privacy?}
\end{quote}
\vspace{-0.5cm}
Answering this question is not a trivial task. As mentioned above, the exponential mechanism is the main approach to ensure DP for auctions. An exponential mechanism firstly creates a probability distribution over all possible auction results such that more preferable result is associated with a higher probability, and then outputs an auction result according to the distribution.  
However, this mechanism can not be directly extended to diffusion auctions as it fails to ensure diffusion IC property. For instance, run the exponential mechanism to the scenario in Figure~\ref{fig:inference} (See Example~\ref{exa:exponential} for a detailed implementation). Assume that all buyers except buyer $b$  reveal their neighbours truthfully. From $b$'s perspective, revealing her neighbour $f$ means getting a lower probability of winning the auction, as the exponential mechanism would distribute the winning probabilities over $7$ buyers instead of $5$. Therefore, the buyers are not incentivised to diffuse auction information to their friends. 

\begin{figure}
    \centering
    \includegraphics[width=0.26\textwidth]{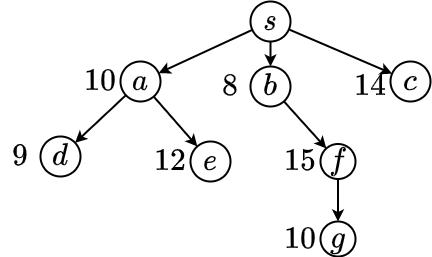}
    \caption{\small A social network with a seller $s$ and seven buyers. 
    The number beside each node is the valuations of the buyer. The seller $s$ has an item to sell, and initially knows only $a,b,c$. 
    The mechanism will construct a probability distribution over potential buyers which determines how likely a buyer is to win the item.}
    \label{fig:inference}
\end{figure}

\smallskip

\noindent {\bf Contribution.} In this paper, we design a DPDM for the {\em single-unit auction} case where a single seller sells one indivisible item to multiple potential buyers. The seller and the buyers are assumed to be nodes in a social network with their connections represented as 
edges. The seller initially only has access to her direct neighbours, and must incentivise the buyers to truthfully report both their valuations of the item, and their neighbourhood. 
%
At the same time, the DPDM should ensure the DP property for buyers' bids. 
We design two DPDMs: {\bf recursive DPDM} and {\bf layered DPDM}. The idea for these two mechanisms is {\em market division} that partitions the buyers into sub-markets. The mechanism then associates a probability with each sub-market. To ensure diffusion IC, the probability should be monotonic on the size of the sub-markets:
\vspace{-0.3cm}
\begin{itemize}[leftmargin=*]
    \item The recursive DPDM maps the network into a tree that captures information flow among buyers. Then it recursively divides the market such that each sub-tree is a sub-market and its probability is non-decreasing on the size of the sub-tree. 
    \item The layered DPDM also relies on the tree above, except the market is not partitioned by sub-trees, but rather by buyers' distances from the seller. In this way, each layer is a sub-market and its probability is fixed. 
\end{itemize}
\vspace{-0.3cm}
These two mechanisms are proven to meet  all the desirable properties. The layered DPDM has a lower bound on expected social welfare. The recursive DPDM achieves a better social welfare empirically. We demonstrate this using a series of experiments that simulate diffusion auctions over three real-world social network datasets. 
Our experiments reveal that in most cases, the recursive DPDM reaches comparable social welfare as the theoretical upper bound. We now highlight our contributions: 

\begin{enumerate}[leftmargin=*]
    \item We  expand diffusion mechanisms adding the DP condition. This builds a bridge between diffusion auctions and privacy preservation. See Section \ref{sec:privacy}.  
    \item Using the idea of market division, we design  recursive DPDM (Section~\ref{sec:recursive}) and layered DPDM (Section~\ref{sec:layered}). These mechanisms are IC and differentially private. 
    \item We empirically evaluate our two mechnaisms 
    on real-world network datasets. See Section~\ref{sec:experiment}.
\end{enumerate}

	
\vspace{-0.3cm}

\section{Related work}

\paragraph*{\bf Differentially private mechanism.}
Differential privacy (DP) is proposed to protect individual data from inference attacks on an aggregate query over a database \citep{dwork2006calibrating}. The notion has since been extended to various domain such as statistical data inference \citep{dwork2008differential}, decision trees \citep{fletcher2019decision}, and unstructured data \citep{zhao2022survey}. \citet{mcsherry2007mechanism} extend DP to auctions and propose exponential mechanism.
This mechanism ensures a weaker version of IC, namely approximate IC, which ensures that any user can only gain a bounded extra utility from misreporting. This solution concept is adopted in subsequent studies \citep{zhu2014differentially} and \citep{Diana2020Differentially} on multi-item auctions and one-shot double auctions. 
As approximate IC allows bidders to have non-zero incentives to lie, these methods would not meet the requirements in our problem.  

Many works design DP auctions that ensure traditional version of IC \citep{Huang2012TheEM,Xiao2013IsPC,zhu2015differentially,Lin2018Frameworks}. Specifically, \citet{Huang2012TheEM,Xiao2013IsPC} propose general methods to transform a classical IC mechanism to a privacy preserving counterpart that is still IC. The method of \citep{Xiao2013IsPC} works only when the valuation space is small and can not be applied to general problems, including ours. In contrast, the transformation method in \citep{Huang2012TheEM} can be applied to more general problems. The transformed mechanism can be seen as a generalisation of Vickrey-Clarke-Groves (VCG) mechanism \citep{groves1973incentives}, which is paired with a carefully designed payment rule. However, when the mechanism is applied to multi-item auctions, it is approximately IC rather than IC. Later, \citet{zhu2015differentially} and \citet{Xu2017PADS,Lin2018Frameworks} propose mechanisms that combine the exponential mechanism with the payment rule in \citep{Archer2001TruthfulMF}, applying to combinatorial auctions and reverse auctions, respectively. 

No mechanism above can be applied to our problem of designing DPDM because they fail to ensure diffusion IC. We next introduce existing  diffusion auction mechanisms.  

\paragraph*{\bf Diffusion auction.} Diffusion auction is an emerging topic in mechanism design. 
\citet{li2017mechanism} are the first to investigate diffusion auction and propose information diffusion mechanism (IDM), a mechanism for single-unit auction in a social network. The basic idea is to give monetary reward to buyers who are critical to diffusion, and it ensures diffusion IC. Following this idea, \citet{li2019diffusion,zhang2020incentivize,zhang2020redistribution} further study single-unit diffusion auction from different aspects. 
Later, \citet{zhao2018selling,kawasaki2020strategy} extend single-unit diffusion auctions to multiple-unit cases and propose generalised IDM (GIDM) and DNA-MU, resp. 
However, all of these mechanisms are deterministic and suffer from privacy leakage risks. 
%

\section{Problem formulation}
\label{sec:problem}

\subsection{Preliminaries}
Consider the following setup: There is a seller, denoted by $s$, and $n$ buyers, denoted by $N=\{1,2,3,...,n\}$. Seller $s$ has a single indivisible item to sell. Each buyer $i\in N$ is willing to buy the item and attaches a {\em valuation} $v_i$ to the item. Valuation $v_i$ is the maximum amount of money that $i$ is willing to pay. 
This value is private to the seller. 

The seller and the buyers form a social network, represented by a graph $G=(V,E)$, where $V=N \cup \{s\}$ is the vertex set and $E\subseteq V^2$ is the edge set. Each node $i\in V$ has a neighbour set, denoted by $r_i\coloneqq \{j\in V \mid (i,j)\in E\}$. 
We assume that only the seller's neighbours know the auction information initially. The seller would like to attract more buyers to participate in the auction and spread the auction information. Each buyer is able to deliver the auction information to her neighbours. The set $r_i$ is also a private information of buyer $i$. 

Each buyer $i \in N$, once informed with the auction, can participate in the auction. Also, for each buyer a pair consisting of her valuation and neighbour set is called the {\em profile} of the buyer. We use $\theta_i\coloneqq (v_i,r_i)$ to denote this profile. 
The profile is known to the buyer and it is hidden to anyone else. Let $\Theta$ denote the set of all possible profiles. Also, we let $\theta \coloneqq (\theta_1,\ldots,\theta_n)$ be the {\em global profile} of all buyers and $\theta_{-i} \coloneqq (\theta_1,\ldots \theta_{i-1},\theta_{i+1},\ldots \theta_n)$ be the global profile of all buyers except for $i$. 
In the auction, each buyer is asked to report her profile $\theta_i'=(v_i',r_i')$, which is not necessarily the true one. We define $\theta'\in \Theta^n$ as the reported global profile of all buyers. 
Given $\theta'$, we construct a directed graph $G_{\theta'}=(V_{\theta'},E_{\theta'})$: add a directed edge $(i,j)$ if $j$ is reported by $i$ as a neighbour. We call such graph {\em profile digraph}. 

Diffusion auction has two forms of information asymmetry:
    {\bf (1) Valuation asymmetry.} The buyers' true valuations are private information and hidden from the seller. Thus buyers have an advantage over the seller as they can misreport their valuations. The auction should prevent misreporting of valuation through appropriate allocation and pricing strategies. 
    {\bf (2) Neighbourhood asymmetry.} By Bulow-Klemperer theorem, the revenue of an auction increases as the number of buyers grows \citep{bulow1996auctions}. However, as buyers' neighbours on the social network are hidden, the seller would hope the buyers to diffuse the auction information to their neighbours to allow more participants to join. However, being rational, the buyers are not necessarily willing to disseminate the auction information as this may hinder their own chance of winning. 
    Here, we follow the standard convention 
    and assume that the reported neighbour set $r_i'$ is a subset of $r_i$.

Diffusion auction mechanisms are designed to address these two challenges. Now we give the definition of a mechanism. A mechanism, denoted by $M$, takes the reported global profile $\theta'$ of all buyers as input, and determines who is allocated the item and how much to pay.  

\begin{definition}
    A {\em mechanism} $M$ consists of two functions $(\pi(\cdot),$ $p(\cdot))$, where $\pi \colon \Theta^n \to \{0,1\}^n$ is an {\em allocation function} and $p \colon \Theta^n \to \R^n$ is a {\em payment function}. 
\end{definition}

The allocation function determines whether the buyers get the item while the payment function determines the amount of money that the buyers need to pay.
Given a reported global profile $\theta'$ of all buyers, we write the {\em allocation result} $\pi(\theta')$ as $(\pi_1(\theta'),\ldots,$ $\pi_n(\theta'))$ and the {\em payment result} $p(\theta')$ as $(p_1(\theta'),\ldots,p_n(\theta'))$, where $\pi_i(\theta')$ and $p_i(\theta')$ is buyer $i$'s allocation and payment. The {\em utility} of buyer $i$ with profile $\theta_i=(v_i,r_i)$ is $u_i (\theta') = v_i\pi_i(\theta') - p_i(\theta')$ when reported global profile is $\theta'$. 
The {\em social welfare} of mechanism $M$ on $\theta'$, denoted by $sw_M(\theta')$, is defined as the sum of the seller and the buyers' utility, i.e., $sw_M(\theta')=\sum_{i\in V } u_i(\theta')$. We aim to maximise the social welfare. 




\subsection{Privacy-aware diffusion auction}
\label{sec:privacy}

In addition to Challenges (1) and (2) above, we consider a third challenge in diffusion auction when the buyers are privacy-aware. 
    {\bf (3) Valuation privacy.} 
    Once the auction result is annouced, an attacker with certain background information may infer the bid information from the published auction result. 
    This is known as the {\em inference attack} \citep{li2017mechanism}. 
    This disadvantages the buyer(s) whose private valuation is diclosed.
    Therefore, the buyers require the guarantees that their private valuations are protected. 


To achieve privacy preservation, we apply a randomised mechanism $M$ to implement an auction on the reported global profile. 

\begin{definition}
A {\em randomised mechanism $M$} is one that, given a global profile $\theta$, outputs a pair $(\pi,p)$ such that $\pi$ is a randomised allocation function and $p$ is a randomised payment function.
\end{definition}
 
\noindent Given a global profile $\theta$, the randomised mechanism $M$ outputs $\pi(\theta)$ and $p(\theta)$ such that $\pi(\theta)$ is a random variable with possible values $\{0,1\}^n$ and $p(\theta)$ is a random variable with possible values $(\R^+)^n$. 
%
We use the concept of differential privacy to define the privacy protection of a mechanism. Basically, differential privacy requires that the distributions over the outcomes are nearly identical when the global profiles are nearly identical. The privacy protection level is measured by a privacy parameter $\epsilon \in \R^+$. 

\begin{definition}
A randomised mechanism $M$ is {\em  $\epsilon$-differential privacy} ($\epsilon$-DP) if for any two global profiles $\theta,\theta'\in \Theta^n$ that differ on a single buyer's valuation, and for any  possible outcome $o \in O$,
\begin{equation}
\label{eqn:DP}
    \Pr[M(\theta)= o]\leq \exp(\epsilon)\Pr[M(\theta')= o]
\end{equation}
\end{definition}

Eqn.~(\ref{eqn:DP}) shows if any buyer $i$ changes her reported profile from  $\theta_i=(v_i,r_i)$ to $\theta_i'=(v_i',r_i)$, the auction outcome does not change too much. Therefore, no one could infer the valuation of any buyer from the randomised outcome. 

{\em Exponential mechanism} \citep{mcsherry2007mechanism} is an existing mechanism that ensures $\epsilon$-DP for valuation privacy. 
Given a global profile, an exponential mechanism creates a distribution over all possible auction outcomes, and outputs an outcome according to the distribution. Intuitively, the higher a reported valuation is, the more likely the corresponding buyer is selected as a winner. 
Specially, given a global profile $\theta$, define a {\em score function} $\sigma \colon \Theta^n \times O \to \R$ that assigns a real valued score to each pair $(\theta,o)$ from $\Theta^n \times O$. The more preferable an outcome is, the higher the score of the outcome is. An exponential mechanism $M(\theta)$ outputs a result $o^* \in O$ with probability 
\[
\frac{\exp(\epsilon \sigma(\theta,o^*))}{\sum_{o \in O} \exp(\epsilon \sigma(\theta,o))}
\]
In our problem, a result corresponds to that a certain buyer $i$ wins, and we use $o_i$ to denote this result.   

In randomised mechanisms, we assume that the buyers are risk-neutral and care about their utilities in expectation. We use $\E_{M}[u_i(\cdot)]$ to denote $i$'s expected utility in $M$ and redefine the standard IC and IR properties by expected utility. 

\begin{definition}
Let $M$ be a randomised mechanism,
\begin{itemize}[leftmargin=*]
    \item The mechanism $M$ is IC if for all $i\in N$, all $\theta_i, \theta_i' \in \Theta$ and for all $\theta_{-i}',\theta_{-i}''\in \Theta^{n-1}$, we have the following,  $\E_{M}[u_i((\theta_i,\theta_{-i}'))] \geq \E_{M}[u_i((\theta_i',\theta_{-i}''))].$ 
    \item The mechanism $M$ is IR if for all $i \in N$ and all $\theta_{-i}'\in \Theta^{n-1}$, we have 
    $\E_{M}[u_i((\theta_i,\theta_{-i}'))]\geq 0.$
\end{itemize}
\end{definition}
The IR and IC properties ensure that buyers are willing to participate in the auction and to reveal their true valuations and neighbours, as they are rational and doing so leads to the best expected utilities. Hence, information asymmetry issues can be addressed. 

\noindent The social welfare of $M$ is also in expectation, i.e.,  $$\E_{M}[sw_M(\theta)]=\sum_{i\in V} \E_{M}[u_i(\theta)].$$ 

We aim to design a randomised mechanism that is IC, IR, $\epsilon$-DP (for reasonable $\epsilon$) while maximising social welfare. 

\section{Recursive DPDM}
\label{sec:recursive}

Preserving valuation privacy in diffusion auctions is not a trivial task. On one hand, existing diffusion auctions, including IDM \citep{li2017mechanism}, CMD \citep{li2019diffusion}, and FDM \citep{zhang2020incentivize}, are deterministic, and thus fail to preserve privacy. On the other hand, existing differential private mechanisms, including exponential mechanism, fail to incentivise truthful report of neighbours, which is illustrated in Example~\ref{exa:exponential}. 

\begin{example}
\label{exa:exponential}
We apply exponential mechanism paired with score function $\sigma(\theta,o_i)=v'_i$ to the scenario in Figure~\ref{fig:inference}. That is, the score of the result that $i$ wins is $i$'s reported valuation $v_i'$. We assume that the buyers truthfully report their valuations. Then buyer $i$ wins with probability $\exp(\epsilon v_i )/\sum_{\kappa\in N}\exp(\epsilon v_{\kappa})$.
Now if buyer $b$ reports her neighbour $f$, $b$ wins with probability $\exp(8\epsilon)/\sum_{\kappa\in N}\exp(\epsilon v_\kappa)$, whereas she wins with probability $\exp(8\epsilon)/\sum_{\kappa\in N \setminus \{f,j\}}\exp(\epsilon v_\kappa)$ had she chose not to report $f$. In the latter case, the winning probability is even higher, and thus $b$ has incentive to hide her neighbours.
\end{example}

To incentivise buyers to diffuse auction information, we need to ensure each buyer's utility of reporting her neighbours should be no less than that of non-reporting. 
We now propose {\em recursive DPDM} $\Mrec$ to achieve this condition. The basic idea is ``market division'', i.e., treat the social network as a market, partition the market into multiple sub-markets and assign each sub-market a probability with which buyers in this sub-market win, as shown in Eqn.~\eqref{eqn:subtree}. In this case, each buyer would report as many neighbours as possible in order to maximise the probability of the sub-market she belongs to. Then the buyers in a sub-market share the probability of the sub-market in such a way that the winning probability of any buyer is independent from her children, as shown in Eqn.~\eqref{eqn:node}. Therefore, the buyers have no competition with their children and have no incentive to misreport them.  

We now describe $\Mrec$ in detail: Fix a score function $\sigma(\cdot)$ that is non-decreasing in reported valuation $v_i'$. Given a reported global profile $\theta'$, a privacy parameter $\epsilon$ and the score function $\sigma(\cdot)$ as input, $\Mrec$ works as follows: 

\noindent {\bf (1) Construction of diffusion critical tree.} Given a profile digraph $G_{\theta'}$, $\Mrec$ first constructs a {\em diffusion critical tree}, denoted by $T_{\theta'}$. When the context is clear, we write the tree as $T$. The idea of diffusion critical tree is originally 
introduced by \citep{zhao2018selling}. For any buyers $i,j \in V_{\theta'}$, we say that $i$ is {\em $\theta'$-critical} to $j$, denoted by $i \preceq_{\theta'} j$, if all paths from $s$ to $j$ in $G_{\theta'}$ go through $i$. A diffusion critical tree is a rooted tree, where the root is seller $s$ and the nodes $V_{\theta'}$ are the buyers who are connected to $s$, and for each $j\in V_{\theta'}$, her parent is the node $i \preceq_{\theta'} j$ who has the closest distance to $j$. When there are more than one parents, only one node is randomly selected as the parent. 
The {\em depth} of buyer $i$, denoted by $d_i$, is the distance from $s$ to $i$. 

\noindent {\bf (2) Assignment of winning probabilities. } This step determines the probabilities that buyers win the item. 
This is a recursive process. This process starts with the constructed 
$T$ rooted by $s$. 
Given a (sub-)tree rooted by $i \in V$, $\Mrec$ assigns a probability to each sub-tree rooted by $j\in r_i$, and a winning probability to each $j\in r_i$. This operation is repeated for $j$'s children, children of $j$'s children and so on until there is no more children. 


\noindent {\em (a) Assignment of probabilities to sub-trees.} Let $T[i]$ denote the sub-tree rooted by $i$. $T[i]$ consists of node $i$ and all of $i$'s descendants. Let $T(i)$ denote $T[i]$ with $i$ removed, i.e., $T(i)\coloneqq T[i] \setminus \{i\}$.
Given a sub-tree $T[i]$, $\Mrec$ divides the market in $T[i]$ to $|r_i|+1$ sub-markets, one for $i$ and each of the other for a sub-tree $T[j]$, where $j\in r_i$. Then $\Mrec$ assigns a probability $\Pr_i^{\theta'}(\theta_i')$ to $i$ with $\theta_i'$ and $\Pr^{\theta'}_{T[j]}$ to each $T[j]$, where $j\in r_i$. When the context is clear, we write $\Pr_i$ and $\Pr_{T[j]}$ for  $\Pr_i^{\theta'}(\theta_i')$ and $\Pr^{\theta'}_{T[j]}$, respectively. 
We define $\Pr_i$ later in Step (2).b. 
For notational convenience, given a set of nodes $S\subseteq T$, we let $\EXP(S)$ be the sum  
\begin{equation*}\label{eqn:EXP}
\EXP(S) = \sum_{\kappa\in S} \exp(\epsilon\sigma(\theta',o_\kappa)).
\end{equation*}

Now we define $\Pr_{T[j]}$ for each $j\in r_i$ as
\begin{equation}
\label{eqn:subtree}
\Pr_{T[j]}=\left(\Pr_{T[i]}-\Pr_i\right)\times\frac{\EXP(T[j])}{\EXP(T(i))}    
\end{equation}

\noindent {\em (b) Assignment of winning probabilities to buyers within a sub-market.} In a sub-tree $T[i]$, $\Mrec$ assigns the winning probability $\Pr_j$ to each $j\in r_i$ as
\begin{equation}
\label{eqn:node}
\Pr_j = \left(\Pr_{T[i]}-\Pr_i\right) 
\times\frac{\EXP(j)}{\EXP( T(i) \setminus T(j))}    
\end{equation}
\noindent At the very beginning, $\Mrec$ starts with the tree $T$ rooted by $s$. We label $s$ as node $0$ and set $\Pr_{T[0]}=1$ and $\Pr_0=0$. $\Mrec$ ends with the leaves. For a sub-tree $T[i]$ where each $j\in r_i$ are leaves, $\Mrec$ assigns the winning probability to each $j$ as $\Pr_j=\left(\Pr_{T[i]}-\Pr_i\right)\times \frac{\EXP(j)}{\EXP(T(i))}$.

\noindent {\bf (3) Allocation and payment.} Randomly select a buyer $w$ as a winner according to the constructed distribution in Step (2). Set $w$'s allocation $\pi_{w}=1$, and payment as 
\begin{equation}
\label{eqn:payment}
p_{w}=v_{w}'-\int_{0}^{v_{w}'}\Pr_w((x,r_w')) d x / \Pr_w(\theta_w')
\end{equation}

We present the details of $\Mrec$ in Algorithm~\ref{alg:recursive} and give a running example of Step (2) in Example~\ref{exa:recursive}. 

\begin{algorithm}[H]
	\caption{\small Recursive DPDM $\Mrec$}
	\label{alg:recursive}
	\begin{algorithmic}[1]
		\Require Reported global profile $\theta'$, privacy parameter $\epsilon$ and score function $\sigma$
		\Ensure Allocation result $\pi(\theta')$ and payment result $p(\theta')$
		\State Initialise $\pi(\theta')= \textbf{0}, p(\theta')= \textbf{0}$ 
		\State Construct a profile digraph $G_{\theta'}=(V_{\theta'},E_{\theta'})$
		\State Construct a critical diffusion tree $T_{\theta'}$
		\State Run GetPro($T_{\theta'}[0],1,0$)
	    \State Randomly select a buyer $w$ with the distribution
	    \State Set $\pi_{w}=1$ and $p_w$ by Equation~(\ref{eqn:payment})
	\end{algorithmic}
\end{algorithm}

\begin{algorithm}[H]
	\caption{\small GetPro}
	\label{alg:subtree}
	\begin{algorithmic}[1]
	\Require (Sub-)Tree $T[i]$, probabilities $\Pr_{T[i]}$ and $\Pr_i$
	\Ensure Probabilities $\Pr_{T[j]}$ and $\Pr_j$, $j\in r_i$ 
	\For{$j\in r_i$}
	\State Calculate $\Pr_{T[j]}$ of sub-tree $T[j]$ by Equation~(\ref{eqn:subtree}) 
	\State Calculate $\Pr_j$ of buyer $j$ by Equation~(\ref{eqn:node}) 
	\State Run GetPro($T[j],\Pr_{T[j]},\Pr_j$)
	\EndFor
	\end{algorithmic}
\end{algorithm}

\begin{example}
\label{exa:recursive}
We apply $\Mrec$ paired with score function $\sigma(\theta,o_i)=v'_i$ to the scenario in  Fig.~\ref{fig:inference}. 
Firstly, $\Pr[T]=1$ and $\Pr_s=0$. Next we calculate the probabilities of $s$'s children. The probability for $T[a]$ is $\Pr_{T[a]}=(\exp(10\epsilon)+\exp(9\epsilon)+\exp(12\epsilon))/{\EXP(T)}$.
Buyer $a$ wins with probability $\Pr_a(10)=\exp(10\epsilon)/(\EXP(T)-(\exp(9\epsilon)+\exp(12\epsilon)))$. Similarly, we can get the probabilities for $T[b], T[c]$ and $b,c$. 
Consider buyer $d$. $d$ wins with probability $\Pr_d(9)=(\Pr(T[a])-\Pr_a)\times \exp(9\epsilon)/(\exp(9\epsilon)+\exp(12\epsilon))$. Similarly, we can also get the probabilities for $e,f,g$. 
\end{example}

Next we show that recursive DPDM satisfies IC, IR and DP. The next classical result is important for IC.

\begin{theorem}[\citep{Archer2001TruthfulMF}]
\label{thm:IC}
Let $\Pr_i(v_i')$ be the probability that $i$ wins when she reports $v_i'$. A mechanism $M=(\pi, p)$ is incentive compatible in terms of valuations if and only if, for any $i\in N$, 
\begin{enumerate}
    \item $\Pr_i(v_i')$ is monotonically non-decreasing in $v_i'$;
    \item $\E[p_i]=v_i\Pr_i(v_i')-\int_{0}^{v_{i}'}\Pr_i(x)dx$ 
\end{enumerate}
\end{theorem}

\begin{lemma} 
\label{lem:recursiveIC}
	Recursive DPDM $\Mrec$ is incentive compatible in terms of both valuations and neighbours. 
\end{lemma}

\begin{proof}
We first show $\Mrec$ is IC in terms of valuations. By Equation~(\ref{eqn:subtree}), the probability for any sub-tree $T[i]$ is proportional to the score, which is non-decreasing in $v_i'$. Hence, $\Pr_{T[i]}$ in non-decreasing in $v_i'$. Similarly, by Equation~(\ref{eqn:node}), given a sub-tree $T[i]$, the winning probability $\Pr_i$ is non-decreasing in $v_i'$, which meets  the condition (1) in Thm.~\ref{thm:IC}. 
Also, by Equation~(\ref{eqn:payment}), the expected payment $$\E[p_i]=p_i \times \Pr_i= v_{i}' \Pr_i(\theta_i')-\int_{0}^{v_{i}'}\Pr_i((x,r_i'))dx,$$ which meets the condition (2) in Theorem~\ref{thm:IC} when $r_i'$ is fixed. 
Therefore, $\Mrec$ is IC in terms of valuations. 

Next we show $\Mrec$ is IC in terms of neighbours. By the definitions of expected utility and payment function~(\ref{eqn:payment}), we know that $i$'s expected utility is only determined by the winning probability $\Pr_i$. Let  $a^\ell$ be an ancestor of $i$ with distance $\ell$. When $i$ reports truthfully as $\theta_{i}$ and the reported global profile is $\theta'_{-i}$, 
then $i$'s winning probability is
\begin{equation}
\label{eqn:IC}
\begin{aligned}
\Pr_i=&\frac{\EXP(i)}{\EXP(T(a^1)\setminus T(i))} \times (\Pr_{T[a^1]}-\Pr_{a^1}) \\
=& \frac{\EXP(i)}{\EXP(T(a^1)\setminus T(i))} \times \left( \Pr_{T[a^2]}-\Pr_{a^2} \right) \\
&\text{\qquad } \times \left(\frac{\EXP(T[a^1])}{\EXP(T(a^2))}-\frac{\EXP(a^1)}{\EXP( T(a^2) \setminus T(a^1))} \right)\\
=& \frac{\EXP(i)}{\EXP(T(a^1)\setminus T(i))} \times \left( \Pr_T-\Pr_{s} \right) \\
&\times \prod_{\ell=1}^{d_{i}-1}\left(\frac{\EXP(T[a^\ell])}{\EXP(T(a^{\ell+1}))}-\frac{\EXP(a^\ell)}{\EXP( T(a^{\ell+1}) \setminus T(a^\ell))} \right)
\end{aligned}
\end{equation}
If $i$ hides some of her neighbours and reports any $\theta_i'$ where $r_i' \subseteq r_i$, instead, and the others report $\theta_{-i}'$. 
Then in Equation~(\ref{eqn:IC}), $\Pr_T$, $\Pr_{s}$ and $\frac{\EXP(i)}{\EXP(T(a^1)\setminus T(i))}$ does not change. Also, for each $\ell$, $\frac{\EXP(a^\ell)}{\EXP( T(a^\ell) \setminus T(a^{\ell+1}))}$ remains intact, but $\frac{\EXP(T[a^\ell])}{\EXP(T(a^{\ell+1}))}$ decreases. So we can know that $\Pr_i$ decreases when $i$ misreports her neighbourhood. Therefore, we have $\E_{\Mrec}[u_i(((v_i,r_i),\theta_{-i}'))] \geq \E_{\Mrec}[u_i(((v_i,r_i'),\theta_{-i}''))].$
\end{proof}


\begin{lemma}
\label{lem:recursiveIR}
Recursive DPDM $\Mrec$ is individually rational in terms of both valuations and neighbours. 
\end{lemma}
\begin{proof}
Given a global profile $\theta$, for each buyer $i$ with $(v_i,r_i)$, 
$\E_{\Mrec}[u_i(\theta)]=(v_i-p_i(\theta)) \Pr_i(\theta_i) = \int_{0}^{v_{i}}\Pr_i((x,r_i))dx \geq 0.$
Therefore, the lemma holds.  
\end{proof}

In following lemma, we use the following terminologies:
\vspace{-0.3cm}
\begin{itemize}[leftmargin=*]
    \item $d_{\max}$ denotes the maximum depth of the diffusion critical tree, 
    \item $\Delta\sigma$ denotes the largest possible difference in the score function $\sigma$ when applied to two global profiles that differ only on a single user’s valuation, for all possible outcome $o_i\in O$.
\end{itemize}
\vspace{-0.3cm}
\begin{lemma}
\label{lem:recursiveDP}
Given a reported global profile $\theta'$, recursive DPDM $\Mrec$ is $\epsilon d_{\max} \Delta \sigma$-differential privacy, where
$\epsilon$ is the privacy parameter of $\Mrec$. 
\end{lemma}

\begin{proof}
Given two reported global profiles $\theta$ and $\theta'$ that differ in an arbitrary buyer $i$'s reported valuation such that $i$ reports $v_i$ in $\theta$ and $v_i'$ in $\theta'$, we consider the probabilities that $\Mrec(\theta)$ and $\Mrec(\theta')$ return a winner $w$. 
In a critical diffusion tree $T_{\theta}$, let $d_w$ denote the depth of $w$, $a_w^\ell$ be an ancestor of $w$ with distance $\ell$. Also, let $\EXP^{\theta}(T(a_w^1)-T(w))$ and $\EXP^{\theta'}(T(a_w^1)-T(w))$ denote the value derived from $\theta$ and $\theta'$, respectively. Then by Equation~(\ref{eqn:node}), we have 
\begin{equation*}
\label{eqn:winner}
\begin{aligned}
	\frac{\Pr[\Mrec(\theta)=o_w]}{\Pr[\Mrec(\theta')=o_w]} & =\frac{\frac{\EXP(w)}{\EXP^{\theta}(T(a_w^1)-T(w))}}{\frac{\EXP^{\theta'}(w)}{\EXP^{\theta'}(T(a_w^1)-T(w))}}  \times \frac{\Pr^{\theta}_{T[a_w^1]}-\Pr^{\theta}_{a_w^1}}{\Pr^{\theta'}_{T[a_w^1]}-\Pr^{\theta'}_{a_w^1}}
\end{aligned}
\end{equation*}

We repeatedly replace $\Pr^{\theta}_{T[a_w^{\ell}]}$, $\Pr^{\theta}_{a_w^{\ell}}$,  $\Pr^{\theta'}_{T[a_w^{\ell}]}$, $\Pr^{\theta'}_{a_w^{\ell}}$ by expressions of $a_w^{\ell+1}$ until we get an expression of $s$. For each distance $0 \leq \ell < d_w$, we denote  $\frac{\EXP(T[a_w^{\ell}])
}{\EXP(T(a_w^{\ell+1}))}$ as $A^{\theta}_{\ell}$,
$\frac{\EXP(a_w^{\ell})}{\EXP(T(a_w^{\ell+1})\setminus T(a_w^{\ell}))}$ as $B^{\theta}_{\ell}$. For $\theta'$, we have similar notations as $A^{\theta'}_{\ell}$ and $B^{\theta'}_{\ell}$. Then the above ratio can be written as 
\begin{equation*}
\begin{aligned}
\frac{\Pr[\Mrec(\theta)=o_w]}{\Pr[\Mrec(\theta')=o_w]}= \frac{B_0^{\theta}}{B_0^{\theta'}}\times \prod_{\ell=1}^{d_w-1 }{\frac{A^{\theta}_{\ell}-B^{\theta}_{\ell}}{A^{\theta'}_{\ell}-B^{\theta'}_{\ell}}}
\end{aligned}
\end{equation*}
Next we proof the lemma through that for each $0\leq \ell < d_w$, $\frac{A^{\theta}_{\ell}-B^{\theta}_{\ell}}{A^{\theta'}_{\ell}-B^{\theta'}_{\ell}}$ is bounded by $\exp(\epsilon \Delta \sigma)$. Here we skip the proof for this due to space limitation. See details in {\bf App.~\ref{app:lem:recursiveDP}}.
Then we have \begin{equation*}
\begin{aligned}
\frac{\Pr[\Mrec(\theta)=o_w]}{\Pr[\Mrec(\theta')=o_w]} &\leq 
\exp(\epsilon \Delta \sigma) \times \prod_{\ell=1}^{d_w-1} \exp(\epsilon \Delta \sigma) \\
&\leq  \exp(\epsilon d_w \Delta \sigma) \leq \exp(\epsilon d_{\max} \Delta \sigma)
\end{aligned}
\end{equation*}
\end{proof}

Next theorem easily follows from Lemmas~\ref{lem:recursiveIC}, \ref{lem:recursiveIR} \& \ref{lem:recursiveDP}. 

\begin{theorem}
\label{thm:recursive}
Recursive DPDM $\Mrec$ is IC, IR and $\epsilon d_{\max}\Delta \sigma$-DP. 
\end{theorem}

\section{Layered DPDM}
\label{sec:layered}

Following the same idea of market division, we propose layered DPDM $\Mlay$ in this section. Different from $\Mrec$, $\Mlay$ divides the market by the buyers' distances to the seller. 
Specifically, given a constructed critical diffusion tree, $\Mlay$ allocates a certain probability to each layer of the tree, which will be shared by the buyers on this layer. For any buyer, once she is invited by her parent(s), her layer is fixed. Also, the buyer(s) whom she invites will be on the next layer, and thus has no competition with her. 

$\Mlay$ executes the same operations as in $\Mrec$, where the only difference is in Step (2) ``Assignment of winning probabilities''. 
Below we describe Step (2) of $\Mlay$ in detail:


\noindent {\bf (2) Assignment of winning probabilities.} In this step, given a critical diffusion tree $T_{\theta'}$, $\Mlay$ assigns a probability to each layer of the tree and then assigns a winning probability to buyers on each layer. 

\noindent {\em (a) Assignment of probability to layers.}  Now we give the definition of layer. Given a tree, 
the buyers with the same distance $d_i$ form a {\em layer} of a tree. The distance $d_i \in \{1,\ldots, d_{\max}\}$. 
We use $L_{\ell}$ to denote the set of buyers with distance $\ell$, i.e., $L_{\ell}\coloneqq \{i\mid d_i=\ell\}$. For each layer $L_{\ell}, 1 \leq \ell \leq d_{\max}$, $\Mlay$ assigns a probability, denoted by $\Pr^{\theta'}_{L_{\ell}}$. We write it as $\Pr_{L_{\ell}}$ when there is no ambiguity.
Given an infinite decreasing sequence $\gamma=(\gamma_1,\gamma_2,\ldots)$, where $\sum \gamma_i=1$, we define the probability for layer $L_{\ell}$ as
\begin{equation}
\label{eqn:layer}
    \Pr_{L_{\ell}}= \gamma_{\ell}
\end{equation}

\noindent {\em (b) Assignment of winning probability to the buyers on a layer.} On the $\ell$th layer, $\Mlay$ assigns buyer $i$ with $\theta_i'$ on layer $d_i=\ell$ with probability
\begin{equation}
\label{eqn:node2}
    \Pr_i(\theta_i')= \Pr_{L_{\ell}} \times \frac{\EXP(i)}{\EXP(L_{\ell})}
\end{equation}
Once the probability distribution over all possible outcomes is determined, $\Mlay$ computes the payment and randomly selects a winner $w$, following Step (3) of $\Mrec$.  

The complete process of layered DPDM is shown in Alg.~\ref{alg:layer}. Example~\ref{exa:layer} provides a running example of Step (2). 

\begin{example}
\label{exa:layer}
Apply $\Mlay$ paired with score function $\sigma(\theta,o_i)=v_i'$ and sequence $\gamma=\left\{\frac{1}{2^{\kappa+1}}\right\}_{\kappa\in \mathbb{N}}$ to the scenario in Figure~\ref{fig:inference}. 
Then in this graph, three layers, $L_1=\{a,b,c\}, L_2=\{d,e,f\}$, $L_3=\{g\}$  correspond to probabilities $\frac{1}{2},\frac{1}{4},\frac{1}{8}$, resp. In $L_1$, buyer $a$ wins with probability $\exp(10\epsilon)/(2(\exp(10\epsilon)+\exp(8\epsilon)+\exp(14\epsilon)))$. 
Similarly, we get the probabilities for $b$ and $c$. 
Then in $L_2$, $d$ wins with probability $\exp(9\epsilon)/(4(\exp(9\epsilon)+\exp(12\epsilon)+\exp(15\epsilon)))$. The probabilities for $e,f$ can be obtained in a similar way. 
Lastly, in $L_3$, buyer $g$ wins with probability $\frac{1}{8}$.
\end{example}

\begin{algorithm}[H]
	\caption{\small Layered DPDM $\Mlay$}
	\label{alg:layer}
	\begin{algorithmic}[1]
	\Require Reported global profile $\theta'$, privacy parameter $\epsilon$ and score function $\sigma$
	\Ensure  Allocation result $\pi(\theta')$ and payment result $p(\theta')$
	    \State Initialise $\pi(\theta')=\textbf{0}, p(\theta')=  \textbf{0}$ 
		\State Construct a profile digraph $G_{\theta'}=(V_{\theta'},E_{\theta'})$
		\State Construct a critical diffusion tree $T_{\theta'}$
	    \For{$1 \leq \ell \leq d_{\max}$}
	        \State Calculate the probability of layer $\ell$ by Equation~(\ref{eqn:layer})
	        \For{ $i \in L_{\ell}$}
	            \State Calculate  winning probability $\Pr_i$ by Eqn.~(\ref{eqn:node2})
	        \EndFor
	   	\EndFor
	    \State Randomly select a buyer $w$ with the distribution
	    \State Set $\pi_w=1$ and $p_w$ by Equation~(\ref{eqn:payment})
	\end{algorithmic}
\end{algorithm}


Next we show that layered DPDM $\Mlay$ has the desirable properties, including IC, IR and DP. 

\begin{lemma}
\label{lem:layerIC}
Layered DPDM $\Mlay$ is incentive compatible in terms of both valuations and neighbours. 
\end{lemma}

\begin{proof}
The IC property in terms of valuations can be proved in a similar way for Lemma~\ref{lem:recursiveIC}. What we need to show is $\Pr_i$ is non-decreasing in her reported valuation $v_i'$. By Eqn.~(\ref{eqn:node2}), $\Pr_i((v_i',r_i'))$ is proportional to $\sigma(\theta,o_i)$, which is non-decreasing in $v_i'$.  

Then we show IC in terms of neighbours. For an arbitrary buyer $i$, her expected utility is $\E_{\Mlay}[u_i (\theta)] = (v_i - p_i(\theta))\Pr_i$ when the global profile is $\theta$. We plug in Eqn.~(\ref{eqn:payment}) (\ref{eqn:node2}) into $u_i(\theta)$. Then we can see $\Pr_i$ is determined by $d_i$ and $d_i$ is determined by her ancestors. Therefore, her utility will not be effected if she misreports her neighbours, i.e., $\E_{\Mlay}[u_i(((v_i,r_i'),\theta_{-i}))] = \E_{\Mlay}[u_i(((v_i,r_i),\theta_{-i}'))]$. 
\end{proof}

\begin{lemma}
\label{lem:layerIR}
Layered DPDM $\Mlay$ is individually rational in terms of both valuations and neighbours. 
\end{lemma}
The proof of Lemma~\ref{lem:layerIR} follows the same reasoning as Lemma ~\ref{lem:recursiveIR}. See details in {\bf Appendix C}.

\begin{lemma}
\label{lem:layerDP}
Given a reported global profile $\theta'$, layered DPDM $\Mlay$ is $\epsilon \Delta \sigma$-differential private, where $\epsilon$ is the privacy parameter of $\Mlay$. 
\end{lemma}
Lem.~\ref{lem:layerDP} is proved by showing in Eqn.~\eqref{eqn:node2}, the change on a single buyer's valuation is bounded by $\epsilon \Delta \sigma$. Due to space limit, the proof of Lem.~\ref{lem:layerDP} is deferred to {\bf App.~\ref{app:lem:layerDP}}. 
The next thm. then easily follows from Lem.~\ref{lem:layerIC}, \ref{lem:layerIR} and \ref{lem:layerDP}.

\begin{theorem}
\label{thm:layer}
Layered DPDM  $\Mlay$ is IC, IR and $\epsilon \Delta \sigma$-DP. 
\end{theorem}

Next we analyse the expected social welfare of $\Mlay$. 
We consider a hypothetical scenario where the exponential mechanism is applied to the whole social network where the seller knows all buyers. In this scenario, the auction information is diffused to all buyers without any incentive. We call such a mechanism as {\em exponential mechanism with diffusion (EMD)}. EMD has the optimal expected social welfare than all DPDMs and thus is used as the benchmark.

\begin{theorem}
\label{thm:sw}
Given a global profile $\theta$, the expected social welfare of layered DPDM $\Mlay$ is at least $\gamma_{d_{\max}}{\E_{\EMD}[sw_{\EMD}(\theta)]}$.
\end{theorem}
\begin{proof}
Given a global profile $\theta$, the expected social welfare $\E_{\Mlay}[sw_{\Mlay}(\theta)]$ of $\Mlay$ is
\vspace{-0.1cm}

\begin{equation*}
\begin{aligned}
\sum_{i\in V}{\left(v_i \times \Pr^{\Mlay}_{i}(\theta_i)\right)} &=\sum_{i\in V}{ v_i \frac{ \exp(\epsilon,\sigma(\theta,o_i))}{\sum_{j\in L_{d_i}}{\frac{1}{\gamma_{d_i}}} \exp(\epsilon,\sigma(\theta,o_j))}}\\
&= \gamma_{d_{\max}} \E_{\Mlay}[sw_{\Mlay}(\theta)]
\end{aligned}
\end{equation*}
See full derivation in {\bf Appendix~\ref{app:thm:sw}}.
\end{proof}
The next result is an easy corollary. 
\begin{corollary}
For $\gamma=(\frac{a-1}{a},\frac{a-1}{a^2},\dots)$, where $a>1$, layered DPDM achieves an expected social welfare $\geq \frac{a-1}{a^{d_{\max}}}\E_{\EMD}[sw_{\EMD}(\theta)]$. \qed
\end{corollary}

\section{Experiment}
\label{sec:experiment}
We evaluate the performances of $\Mrec$ and $\Mlay$, in terms of social welfare under different privacy levels and valuations on three real world social network datasets. We also analyse the effect of sequence $\gamma=(\frac{a-1}{a},\frac{a-1}{a^2},\ldots)$ on the performance of $\Mlay$.  For each setup, we run $5000$ times and get average social welfare.

{\bf Dataset.} We use three real world network datasets, including Hamsterster friendships with $1,858$ nodes and $12,534$ edges \citep{Kunegis2013KONECT}, Facebook with $4,039$ nodes and $88,234$ edges  \citep{McAuley2012Learning} and Email-Eu-core network $1,005$ nodes and $25,571$ edges \citep{Yin2017Local}. 
For each dataset, the seller $s$ is randomly selected.

{\bf Valuation.} The network datasets contain no information about buyers' valuations.
We generate random numbers as the valuations. 
We consider two commonly used distributions, normal distribution $v_i \sim \mu(50,10)$ and uniform distribution $v_i \sim U{[0,100]}$. We set the parameters such that the average value are same. Nevertheless, our aim is to reveal the general pattern under different distributions and these patterns are independent from these parameters.  

{\bf Privacy parameter.}
To verify the performance of our mechanisms, we also vary privacy parameter $\epsilon \in \{0.01, 0.05, 0.1, 0.15, 0.2, 0.25, 0.3\}$. 
Lem.~\ref{lem:layerDP} and~\ref{lem:recursiveDP} show that, under the same input $\epsilon$, $\Mlay$ and $\Mrec$ ensure different privacy levels. To see the performance under the same guaranteed privacy, we set the input 
$\epsilon$ as $\{ 0.01, 0.05, 0.1, 0.15, 0.2, 0.25, 0.3\}$ for $\Mrec$ and $\{0.01, 0.05, 0.1, 0.15, 0.2, 0.25, 0.3\} d_{\max}$ for the others.

{\bf Score function.} We use  linear function, $\sigma(\theta,o_i)= v_i$, as the score function. 
The linear score function is widely used in previous DP auctions, e.g., \citep{mcsherry2007mechanism,Xu2017PADS}. 

{\bf Decreasing sequence.} For $\Mlay$, we consider different value of $a\in \{1.25,1.5,2,3\}$ in $\gamma=(\frac{a-1}{a},\frac{a-1}{a^2},\dots)$, and evaluate the impact of $a$ on expected social welfare. 

{\bf Benchmark.} Since there is no existing DPDM that can be applied in our problem, we design two hypothetical benchmarks. {\bf Exponential mechanism without diffusion (EMWD)}: We apply the exponential mechanism only to the seller's neighbours. 
The expected social welfare of EMWD can be seen as the lower bound among all DPDMs. 
{\bf Exponential mechanism diffusion (EMD)}: 
See the description of EMD in Section~\ref{sec:layered}. We also compare with IDM \citep{li2017mechanism} (See {\bf App.~\ref{app:IDM}}), which is not DP, to see how much social welfare is sacrificed to achieve DP.



{\bf Results.} Overall, 
when comparing to IDM, the difference in social welfare of the DPDMs decreases with $\epsilon$ increases. 
Then, among DPDMs, EMD performs best in most cases, followed by $\Mrec$ and $\Mlay$. Particularly, $\Mrec$ performs very well. 
The lines of $\Mrec$ even coincide with those of EMD in some cases, e.g., on Facebook \& Email-Eu-core in Fig.~\ref{fig:all}. The deviation of $\Mrec$ from EMD is at most $2.62\%$. 
$\Mrec$ performs better than the layered counterpart. EMWD returns the worst expected social welfare. 
The reason why $\Mrec$ has better expected social welfare than $\Mlay$ is that in $\Mlay$, a probability of $1-\sum_{\ell=1}^{d_{\max}} \gamma_{\ell}$ is not distributed to any buyer, which means that the seller does not sell the item and the social welfare is $0$ with this probability. 

Next we show the effect of different parameters. 
{\bf (1) Dataset.} As shown in each column of Fig.~\ref{fig:all}, the same pattern can be found for different datasets.
{\bf (2) Privacy parameter. } The expected $sw$ increases with $\epsilon$. The less privacy is required, the less noisy is added, and thus the higher probability of returning a result with good social welfare. 
{\bf (3) Valuation.} The $1$st and the $2$nd row of Fig.~\ref{fig:all} show the results with normal and uniform distributions, resp.. 
Under both distributions, $\Mrec$ performs better than $\Mlay$. 
{\bf (4) Sequence.} Fig.~\ref{fig:layer} shows the average social welfare is best when $a=1.5,2$ for Hamsterster and when $a=2,3$ for Facebook and Email-Eu-core. 
When a buyer $i$ with the highest valuation is on a deeper layer, a smaller $a$ leads to a larger probability for the layer where $i$ is and also a larger probability for $i$. The results verify this argument. In Hamsterster (Facebook, Email-Eu-core), the buyers with the highest valuation are on the $4$th ($3$rd, $2$nd) layer.
{\bf (5) same DP.} Fig.~\ref{fig:sameDP} shows when the realised privacy is large, the avg. social welfare of $\Mrec$ is greater than that of $\Mlay$, while when the realised privacy is small, $\Mlay$ is better. 

\vspace{-0.4cm}




\begin{figure}[H]
    \centering
    \includegraphics[width=\columnwidth]{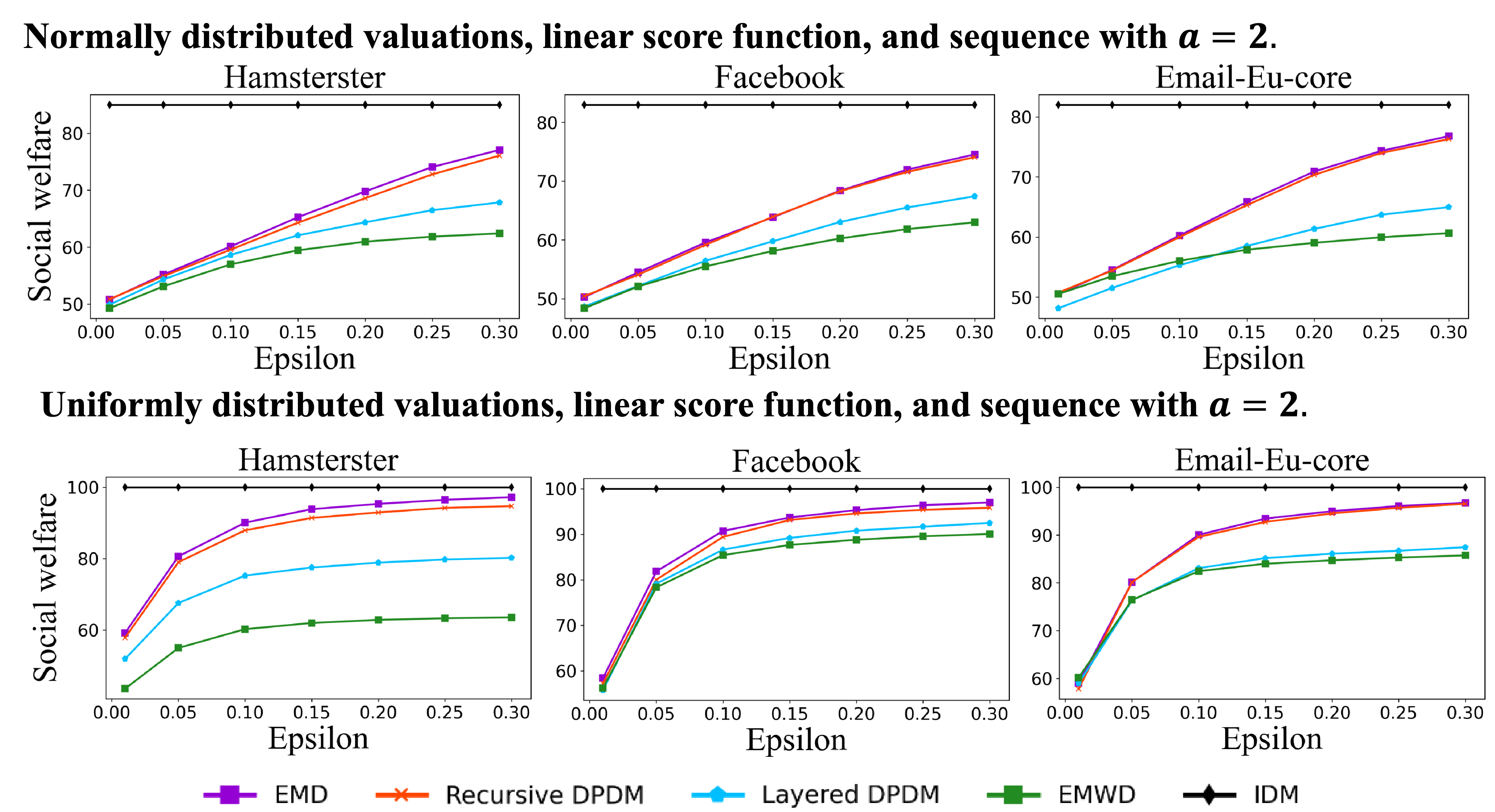}
    \caption{\small Average social welfare of $\Mlay$, $\Mrec$, EMD, EMWD and IDM with different distributions under fixed sequence with $a=2$. Normal distribution is shown in the first row and uniform distribution is shown in the second row.}
    \label{fig:all}
\end{figure}

\vspace{-0.5cm} 
\begin{figure}[H]
    \centering
    \includegraphics[width=\columnwidth]{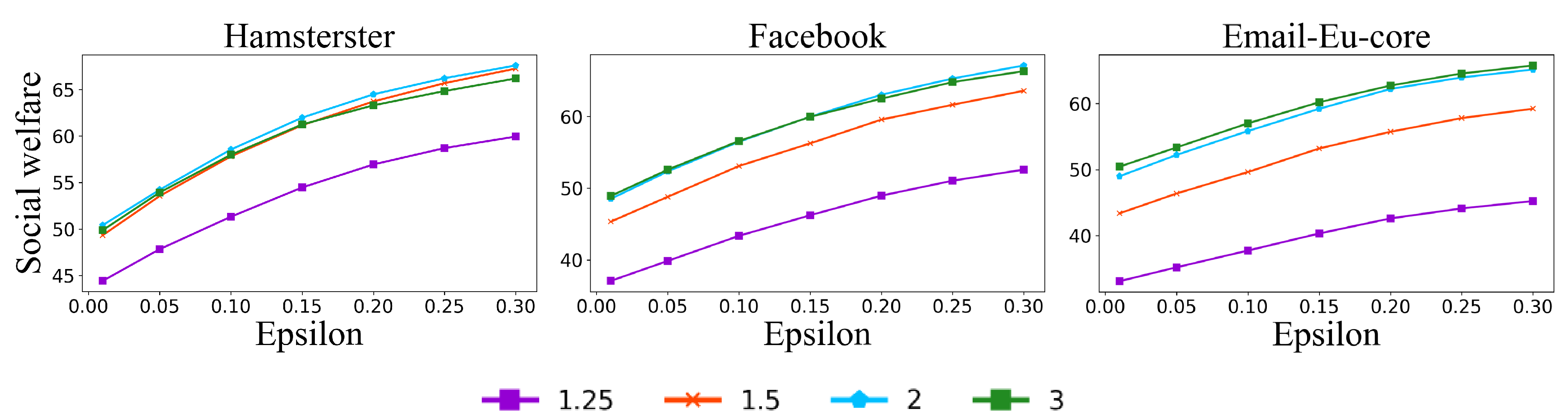}
    \caption{\small Average social welfare of $\Mlay$ with different values of $a$, under normally distributed valuations and linear function.}
    \label{fig:layer}
\end{figure}

\vspace{-0.5cm}
\begin{figure}[H]
    \centering
    \includegraphics[width=\columnwidth]{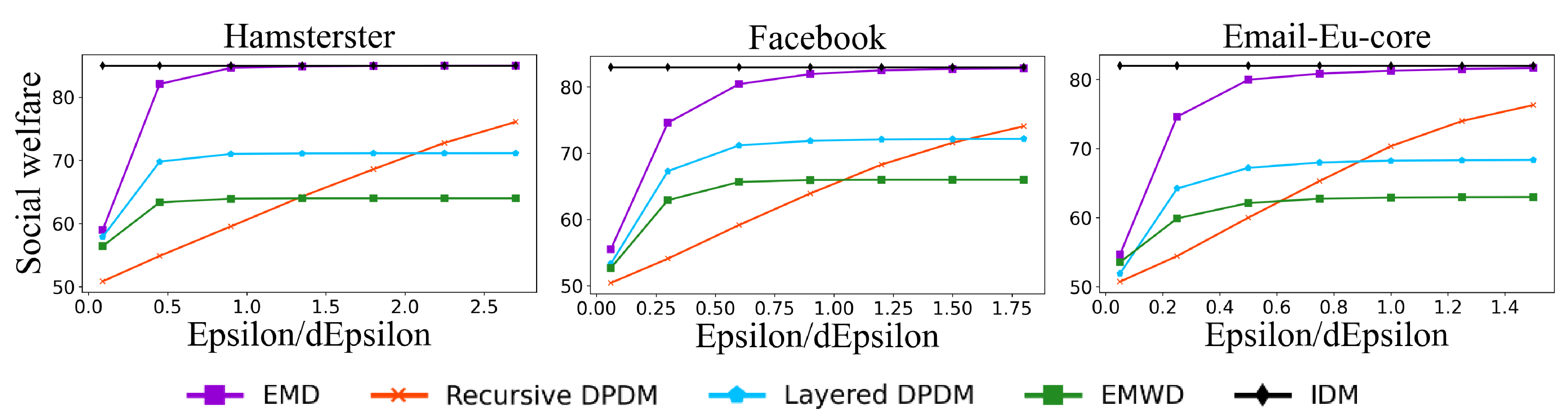}
    \caption{\small Average social welfare of $\Mlay$, $\Mrec$, EMD, EMWD and IDM under normal distribution, linear function and sequence with $a=2$. Horizontal axis represents the value of $\epsilon$ for EMD, EMWD \& $\Mlay$, and $d_{\max}\epsilon$ for $\Mrec$.} 
    \label{fig:sameDP}
\end{figure}

\section{Conclusion and future work}
We consider the problem of designing diffusion auction mechanisms that sells a single item on social networks while preserving valuation privacy. We propose two DPDMs, recursive DPDM and layered DPDM. Also, we theoretically show their incentive and privacy properties and empirically show their good performances in social welfare. We could extend this study by considering the following questions: (1) How to design a DPDM for multi-item auctions? (2) How to design a DPDM that preserves both valuation and neighbourhood privacy? and (3) How to design a DPDM that is group IC where no group of buyers can benefit from joint misreporting?

\bibliography{uai2023-template}

\begin{thebibliography}{25}
\providecommand{\natexlab}[1]{#1}
\providecommand{\url}[1]{\texttt{#1}}
\expandafter\ifx\csname urlstyle\endcsname\relax
  \providecommand{\doi}[1]{doi: #1}\else
  \providecommand{\doi}{doi: \begingroup \urlstyle{rm}\Url}\fi

\bibitem[Archer and Tardos(2001)]{Archer2001TruthfulMF}
Aaron Archer and {\'E}va Tardos.
\newblock Truthful mechanisms for one-parameter agents.
\newblock \emph{Proceedings 2001 IEEE International Conference on Cluster
  Computing}, pages 482--491, 2001.

\bibitem[Bulow and Klemperer(1996)]{bulow1996auctions}
Jeremy Bulow and Paul Klemperer.
\newblock Auctions versus negotiations.
\newblock \emph{The American Economic Review}, 86\penalty0 (1):\penalty0
  180--194, 1996.

\bibitem[Diana et~al.(2020)Diana, Elzayn, Kearns, Roth, Sharifi-Malvajerdi, and
  Ziani]{Diana2020Differentially}
Emily Diana, Hadi Elzayn, Michael Kearns, Aaron Roth, Saeed Sharifi-Malvajerdi,
  and Juba Ziani.
\newblock Differentially private call auctions and market impact.
\newblock In \emph{Proceedings of the 21st ACM Conference on Economics and
  Computation}, page 541–583, 2020.

\bibitem[Dwork(2008)]{dwork2008differential}
Cynthia Dwork.
\newblock Differential privacy: A survey of results.
\newblock In \emph{International conference on theory and applications of
  models of computation}, pages 1--19. Springer, 2008.

\bibitem[Dwork et~al.(2006)Dwork, McSherry, Nissim, and
  Smith]{dwork2006calibrating}
Cynthia Dwork, Frank McSherry, Kobbi Nissim, and Adam Smith.
\newblock Calibrating noise to sensitivity in private data analysis.
\newblock In \emph{Theory of cryptography conference}, pages 265--284.
  Springer, 2006.

\bibitem[Fletcher and Islam(2019)]{fletcher2019decision}
Sam Fletcher and Md~Zahidul Islam.
\newblock Decision tree classification with differential privacy: A survey.
\newblock \emph{ACM Computing Surveys (CSUR)}, 52\penalty0 (4):\penalty0 1--33,
  2019.

\bibitem[Groves(1973)]{groves1973incentives}
Theodore Groves.
\newblock Incentives in teams.
\newblock \emph{Econometrica: Journal of the Econometric Society}, pages
  617--631, 1973.

\bibitem[Huang and Kannan(2012)]{Huang2012TheEM}
Zhiyi Huang and Sampath Kannan.
\newblock The exponential mechanism for social welfare: Private, truthful, and
  nearly optimal.
\newblock \emph{2012 IEEE 53rd Annual Symposium on Foundations of Computer
  Science}, pages 140--149, 2012.

\bibitem[Jian et~al.(2018)Jian, Dejun, Ming, Jia, and
  Guoliang]{Lin2018Frameworks}
Lin Jian, Yang Dejun, Li~Ming, Xu~Jia, and Xue Guoliang.
\newblock Frameworks for privacy-preserving mobile crowdsensing incentive
  mechanisms.
\newblock In \emph{IEEE Trans. Mob. Comput}, page 1851–1864, 2018.

\bibitem[Kawasaki et~al.(2020)Kawasaki, Barrot, Takanashi, Todo, and
  Yokoo]{kawasaki2020strategy}
Takehiro Kawasaki, Nathana{\"e}l Barrot, Seiji Takanashi, Taiki Todo, and
  Makoto Yokoo.
\newblock Strategy-proof and non-wasteful multi-unit auction via social
  network.
\newblock In \emph{Proceedings of the AAAI Conference on Artificial
  Intelligence}, volume~34, pages 2062--2069, 2020.

\bibitem[Kunegis(2013)]{Kunegis2013KONECT}
J\'{e}r\^{o}me Kunegis.
\newblock Konect: The koblenz network collection.
\newblock In \emph{Proceedings of the 22nd International Conference on World
  Wide Web}, WWW '13 Companion, page 1343–1350. Association for Computing
  Machinery, 2013.
\newblock \doi{10.1145/2487788.2488173}.
\newblock URL \url{https://doi.org/10.1145/2487788.2488173}.

\bibitem[Li et~al.(2017)Li, Hao, Zhao, and Zhou]{li2017mechanism}
Bin Li, Dong Hao, Dengji Zhao, and Tao Zhou.
\newblock Mechanism design in social networks.
\newblock In \emph{Thirty-First AAAI Conference on Artificial Intelligence},
  2017.

\bibitem[Li et~al.(2019)Li, Hao, Zhao, and Yokoo]{li2019diffusion}
Bin Li, Dong Hao, Dengji Zhao, and Makoto Yokoo.
\newblock Diffusion and auction on graphs.
\newblock In \emph{Proceedings of the 28th International Joint Conference on
  Artificial Intelligence}, pages 435--441, 2019.

\bibitem[McAuley and Leskovec(2012)]{McAuley2012Learning}
Julian McAuley and Jure Leskovec.
\newblock Learning to discover social circles in ego networks.
\newblock NIPS'12, page 539–547. Curran Associates Inc., 2012.

\bibitem[McSherry and Talwar(2007)]{mcsherry2007mechanism}
Frank McSherry and Kunal Talwar.
\newblock Mechanism design via differential privacy.
\newblock In \emph{48th Annual IEEE Symposium on Foundations of Computer
  Science (FOCS'07)}, pages 94--103. IEEE, 2007.

\bibitem[Ni et~al.(2021)Ni, Chen, Chen, Zhang, Xu, and
  Zhong]{ni2021differentially}
Tianjiao Ni, Zhili Chen, Lin Chen, Shun Zhang, Yan Xu, and Hong Zhong.
\newblock Differentially private combinatorial cloud auction.
\newblock \emph{IEEE Transactions on Cloud Computing}, 2021.

\bibitem[Xiao(2013)]{Xiao2013IsPC}
David Xiao.
\newblock Is privacy compatible with truthfulness?
\newblock \emph{IACR Cryptol. ePrint Arch.}, 2011:\penalty0 5, 2013.

\bibitem[Xu et~al.(2017)Xu, Palanisamy, Tang, and Madhu~Kumar]{Xu2017PADS}
Jinlai Xu, Balaji Palanisamy, Yuzhe Tang, and S.D. Madhu~Kumar.
\newblock Pads: Privacy-preserving auction design for allocating dynamically
  priced cloud resources.
\newblock In \emph{2017 IEEE 3rd International Conference on Collaboration and
  Internet Computing (CIC)}, pages 87--96, 2017.
\newblock \doi{10.1109/CIC.2017.00023}.

\bibitem[Yin et~al.(2017)Yin, Benson, Leskovec, and Gleich]{Yin2017Local}
Hao Yin, Austin~R. Benson, Jure Leskovec, and David~F. Gleich.
\newblock Local higher-order graph clustering.
\newblock In \emph{Proceedings of the 23rd ACM SIGKDD International Conference
  on Knowledge Discovery and Data Mining}, KDD '17, page 555–564. Association
  for Computing Machinery, 2017.
\newblock \doi{10.1145/3097983.3098069}.
\newblock URL \url{https://doi.org/10.1145/3097983.3098069}.

\bibitem[Zhang et~al.(2020{\natexlab{a}})Zhang, Zhao, and
  Chen]{zhang2020redistribution}
Wen Zhang, Dengji Zhao, and Hanyu Chen.
\newblock Redistribution mechanism on networks.
\newblock In \emph{Proceedings of the 19th International Conference on
  Autonomous Agents and MultiAgent Systems}, pages 1620--1628,
  2020{\natexlab{a}}.

\bibitem[Zhang et~al.(2020{\natexlab{b}})Zhang, Zhao, and
  Zhang]{zhang2020incentivize}
Wen Zhang, Dengji Zhao, and Yao Zhang.
\newblock Incentivize diffusion with fair rewards.
\newblock In \emph{ECAI 2020}, pages 251--258. IOS Press, 2020{\natexlab{b}}.

\bibitem[Zhao et~al.(2018)Zhao, Li, Xu, Hao, and Jennings]{zhao2018selling}
Dengji Zhao, Bin Li, Junping Xu, Dong Hao, and Nicholas~R Jennings.
\newblock Selling multiple items via social networks.
\newblock In \emph{Proceedings of the 17th International Conference on
  Autonomous Agents and MultiAgent Systems}, pages 68--76, 2018.

\bibitem[Zhao and Chen(2022)]{zhao2022survey}
Ying Zhao and Jinjun Chen.
\newblock A survey on differential privacy for unstructured data content.
\newblock \emph{ACM Computing Surveys (CSUR)}, 54\penalty0 (10s):\penalty0
  1--28, 2022.

\bibitem[Zhu and Shin(2015)]{zhu2015differentially}
Ruihao Zhu and Kang~G Shin.
\newblock Differentially private and strategy-proof spectrum auction with
  approximate revenue maximization.
\newblock In \emph{2015 IEEE conference on computer communications (INFOCOM)},
  pages 918--926. IEEE, 2015.

\bibitem[Zhu et~al.(2014)Zhu, Li, Wu, Shin, and Chen]{zhu2014differentially}
Ruihao Zhu, Zhijing Li, Fan Wu, Kang Shin, and Guihai Chen.
\newblock Differentially private spectrum auction with approximate revenue
  maximization.
\newblock In \emph{Proceedings of the 15th ACM international symposium on
  mobile ad hoc networking and computing}, pages 185--194, 2014.

\end{thebibliography}

\appendix

\newpage

\noindent {\LARGE \bf Appendix}

\section{IDM}\label{app:IDM}

Here, we introduce the first diffusion auction for selling single item, IDM [13]. A key concept of IDM is diffusion critical sequence. Given a profile digraph $G_{\theta'}$, for any buyers $i,j \in V_{\theta'}$, $i$ is {\em $\theta'$-critical} to $j$, denoted by $i \preceq_{\theta'} j$, if all paths from $s$ to $j$ in $G_{\theta'}$ go through $i$. A {\em diffusion critical sequence} of $i$, denoted by $C_i$, is a sequence of all diffusion critical nodes of $i$ and $i$ itself ordered by $\theta'$-critical relation. That is, $C_i=(x_1,x_2\ldots,x_k,i)$, where $x_1 \preceq_{\theta'} x_2 \preceq_{\theta'}\ldots\preceq_{\theta'} x_k\preceq_{\theta'}i$. Based on this concept, IDM works as follows. IDM first locates the buyer $m$ with the highest valuation among all buyers. Then it allocates the item to the buyer $w$, who has the highest valuation when the buyers after $w$ are not considered. The winner $w$ pays the highest bid without her participation, and each diffusion critical node is rewarded by the increased payment due to her participation. 

\section{Proof of Lemma~\ref{lem:recursiveDP}}
\label{app:lem:recursiveDP}
{\bf Lemma~\ref{lem:recursiveDP}.} {\it Given a reported global profile $\theta'$, recursive DPDM $\Mrec$ is $\epsilon d_{\max} \Delta \sigma$-differential privacy, where
$\epsilon$ is the privacy parameter to $\Mrec$. }

\begin{proof}
Given two reported global profiles $\theta$ and $\theta'$ that differ in an arbitrary buyer $i$'s reported valuation such that $i$ reports $v_i$ in $\theta$ and $v_i'$ in $\theta'$, we consider the probabilities that $M(\theta)$ and $M(\theta')$ return a winner $w$. 
In a critical diffusion tree $T_{\theta}$, let $d_w$ denote the depth of $w$, $a_w^\ell$ be an ancestor of $w$ with distance $\ell$. Also, let $\EXP^{\theta}(T(a_w^1)-T(w))$ and $\EXP^{\theta'}(T(a_w^1)-T(w))$ denote the value derived from $\theta$ and $\theta'$, respectively. Then by Equation~(\ref{eqn:node}), we have 
\begin{equation*}
\label{eqn:winner}
\begin{aligned}
\frac{\Pr[M(\theta)=o_w]}{\Pr[M(\theta')=o_w]} &=\frac{\frac{\EXP(w)}{\EXP^{\theta}(T(a_w^1)-T(w))}}{\frac{\EXP^{\theta'}(w)}{\EXP^{\theta'}(T(a_w^1)-T(w))}} \\
& \qquad \times \frac{\Pr^{\theta}_{T[a_w^1]}-\Pr^{\theta}_{a_w^1}}{\Pr^{\theta'}_{T[a_w^1]}-\Pr^{\theta'}_{a_w^1}}
\end{aligned}
\end{equation*}

We repeatedly replace $\Pr^{\theta}_{T[a_w^{\ell}]}$, $\Pr^{\theta}_{a_w^{\ell}}$,  $\Pr^{\theta'}_{T[a_w^{\ell}]}$, $\Pr^{\theta'}_{a_w^{\ell}}$ by expressions of $a_w^{\ell+1}$ until we get an expression of $s$. For each distance $0 \leq \ell < d_w$, we denote  $\frac{\EXP(T[a_w^{\ell}])
}{\EXP(T(a_w^{\ell+1}))}$ as $A^{\theta}_{\ell}$,
$\frac{\EXP(a_w^{\ell})}{\EXP(T(a_w^{\ell+1})\setminus T(a_w^{\ell}))}$ as $B^{\theta}_{\ell}$. For $\theta'$, we have similar notations as $A^{\theta'}_{\ell}$ and $B^{\theta'}_{\ell}$. Then the above ratio can be written as 
\begin{equation*}
\begin{aligned}
\frac{\Pr[M(\theta)=o_w]}{\Pr[M(\theta')=o_w]}= \frac{B_0^{\theta}}{B_0^{\theta'}}\times \prod_{\ell=1}^{d_w-1 }{\frac{A^{\theta}_{\ell}-B^{\theta}_{\ell}}{A^{\theta'}_{\ell}-B^{\theta'}_{\ell}}}
\end{aligned}
\end{equation*}

Next we show for each $0\leq \ell < d_w$, $\frac{A^{\theta}_{\ell}-B^{\theta}_{\ell}}{A^{\theta'}_{\ell}-B^{\theta'}_{\ell}}$ is bounded by $\exp(\epsilon \Delta \sigma)$. To prove it, we first show for for each $\ell$, $(A^{\theta}_{\ell}-A^{\theta'}_{\ell})\times(B^{\theta}_{\ell}-B^{\theta'}_{\ell})\geq 0$ by cases.\\
(1) When $i\in T[a_w^{\ell}]$, we have $A^{\theta}_{\ell}-A^{\theta'}_{\ell} \leq 0,  B^{\theta}_{\ell}-B^{\theta'}_{\ell}\leq 0$ or $A^{\theta}_{\ell}-A^{\theta'}_{\ell} \geq 0,  B^{\theta}_{\ell}-B^{\theta'}_{\ell}\geq 0$ \\
(2) When $i\in T[a_w^{\ell+1}]\setminus T[a_w^{\ell}]$, then $A^{\theta}_{\ell}-A^{\theta'}_{\ell} \leq 0,  B^{\theta}_{\ell}-B^{\theta'}_{\ell}\leq 0$ or $A^{\theta}_{\ell}-A^{\theta'}_{\ell} \geq 0,  B^{\theta}_{\ell}-B^{\theta'}_{\ell}\geq 0$ \\
(3) When $i \notin T[a_w^{\ell+1}]$, then $A^{\theta}_{\ell}-A^{\theta'}_{\ell}= 0,  B^{\theta}_{\ell}-B^{\theta'}_{\ell}= 0$.

Without loss of generality, we assume that $A^{\theta'}_{\ell}= \alpha_1 A^{\theta}_{\ell}, B^{\theta'}_{\ell}= \alpha_2 B^{\theta}_{\ell}, \alpha_1, \alpha_2 \in \R^+$. Plug in these two equations, and we get
$$\frac{A^{\theta}_{\ell}-B^{\theta}_{\ell}}{A^{\theta'}_{\ell}-B^{\theta'}_{\ell}}  =\frac{A^{\theta}_{\ell}-B^{\theta}_{\ell}}{\alpha_1 A^{\theta}_{\ell}- \alpha_2 B^{\theta}_{\ell}}.$$
Then we consider two cases: \\
(1) When $\alpha_1\geq\alpha_2$, we have  $\frac{A^{\theta}_{\ell}-B^{\theta}_{\ell}}{\alpha_1 A^{\theta}_{\ell}- \alpha_2 B^{\theta}_{\ell}} \leq \frac{A^{\theta}_{\ell}-B^{\theta}_{\ell}}{\alpha_1 A^{\theta}_{\ell}- \alpha_1 B^{\theta}_{\ell}} \leq \frac{1}{\alpha_1}.$ \\
(2) When $\alpha_2\geq\alpha_1$, we have  $\frac{A^{\theta}_{\ell}-B^{\theta}_{\ell}}{\alpha_1 A^{\theta}_{\ell}- \alpha_2 B^{\theta}_{\ell}} \leq \frac{A^{\theta}_{\ell}-B^{\theta}_{\ell}}{\alpha_2 A^{\theta}_{\ell}- \alpha_2 B^{\theta}_{\ell}} \leq \frac{1}{\alpha_2}.$ \\

After that, we show that both $\frac{1}{\alpha_1}$ and $\frac{1}{\alpha_2}$ are bounded by $\exp(\epsilon \Delta \sigma)$ as follows.
By definition of $\alpha_1$, we have 
$\frac{1}{\alpha_1} = \frac{A^{\theta}_{\ell}}{A^{\theta'}_{\ell}}=\frac{\EXP^{\theta}(T[a_w^{\ell}])}{\EXP^{\theta'}(T[a_w^{\ell}])}$
$\times \frac{\EXP^{\theta'}(T(a_w^{\ell+1}))}{\EXP^{\theta}(T(a_w^{\ell+1}))}$.\\
(1) When valuation $v_i'\leq v_i$, the second ratio is at most $1$. Then we have 
\begin{equation*}
\begin{aligned}
\frac{1}{\alpha_1} &= \frac{A^{\theta}_{\ell}}{A^{\theta'}_{\ell}} 
\leq \frac{\EXP^{\theta}(T[a_w^{\ell}])}{\EXP^{\theta'}(T[a_w^{\ell}])}\\
&\leq \frac{\sum_{k\in T[a_w^{\ell}]}\exp(\epsilon \sigma (\theta,o_k))}{\sum_{k\in T[a_w^{\ell}]}\exp(\epsilon (\sigma (\theta,o_k) - \Delta \sigma))}
\leq \exp(\epsilon \Delta \sigma)
\end{aligned}
\end{equation*}
(2) When valuation $v_i'\geq v_i$, the first ratio is at most $1$. We have
\begin{equation*}
\begin{aligned}
\frac{1}{\alpha_1} & = \frac{A^{\theta}_{\ell}}{A^{\theta'}_{\ell}} 
\leq \frac{\EXP^{\theta'}(T(a_w^{\ell+1}))}{\EXP^{\theta}(T(a_w^{\ell+1}))}\\
& \leq \frac{\sum_{k\in T(a_w^{\ell+1})}\exp(\epsilon (\sigma (\theta,o_k)+\Delta \sigma))}{{\sum_{k\in T(a_w^{\ell+1})}\exp(\epsilon \sigma (\theta,o_k))}} 
\leq \exp(\epsilon \Delta \sigma)
\end{aligned}
\end{equation*}
In a similar way, we can show that $\frac{1}{\alpha_2} \leq \exp(\epsilon \Delta \sigma)$. 

Therefore we have \begin{equation*}
\begin{aligned}
\frac{\Pr[M(\theta)=o_w]}{\Pr[M(\theta')=o_w]} &\leq 
\exp(\epsilon \Delta \sigma) \times \prod_{1\leq \ell < d_w} \exp(\epsilon \Delta \sigma) \\
&\leq  \exp(\epsilon d_w \Delta \sigma) \leq \exp(\epsilon d_{\max} \Delta \sigma)
\end{aligned}
\end{equation*}
\end{proof}

\section{Proof of Lemma~\ref{lem:layerIR}}\label{app:lem:layerIR}

{\bf Lemma~\ref{lem:layerIR}. }{\em Layered DPDM $\Mlay$ is individually rational in terms of both valuations and neighbours.}
\begin{proof}
Given a global profile $\theta$, for each buyer $i$ with $(v_i,r_i)$, we have 
\begin{equation*}
    \begin{aligned}
        \E_{\Mlay}[u_i(\theta)] &=(v_i-p_i(\theta)) \Pr_i(\theta_i) \\
        & = \int_{0}^{v_{i}}\Pr_i^{\Mlay}((x,r_i))dx \geq 0.
    \end{aligned}
\end{equation*}
Therefore, the lemma holds. 
\end{proof}

\section{Proof of Lemma~\ref{lem:layerDP}}
\label{app:lem:layerDP}

{\bf Lemma~\ref{lem:layerDP}.} {\em Given a reported global profile $\theta'$, layered DPDM $\Mlay$ is $\epsilon \Delta \sigma$-differential private, where $\epsilon$ is the privacy parameter of $\Mlay$. }

\begin{proof}
Given two reported global profiles $\theta$ and $\theta'$  that differ in an arbitrary buyer $i$'s reported valuation such that $i$ reports $v_i$ in $\theta$ and $v_i'$ in $\theta'$, we consider the probabilities that $M(\theta)$ and $M(\theta')$ return a winner $w$. 

Without loss of generality, we assume that $w$ is in $L_{\ell}$, then we have
\begin{equation*}
   \begin{aligned}
   \frac{\Pr[M(\theta)=o_{w}]}{\Pr[M(\theta')=o_{w}]} &=
   \frac{\Pr_{L_{\ell}} \times \frac{\EXP^{\theta}(w)}{\EXP^{\theta}(L_{\ell})}}{\Pr_{L_{\ell}} \times \frac{\EXP^{\theta'}(w)}{\EXP^{\theta'}(L_{\ell})}} \\
   & = \frac{\EXP^{\theta}(w)}{\EXP^{\theta'}(w)} \frac{\EXP^{\theta'}(L_{\ell})}{\EXP^{\theta}(L_{\ell})}   
   \end{aligned}
\end{equation*}

When $i$ is not on layer $L_{\ell}$, $\frac{\Pr[M(\theta)=o_{w}]}{\Pr[M(\theta')=o_{w}]}=1 \leq \exp(\epsilon \Delta \sigma)$. Otherwise, when $i$ is on layer $L_{\ell}$, we consider two cases. \\
(1) $v_i < v_i'$. As $\sigma(\cdot)$ is non-decreasing in $v_i$, the first ratio is at most $1$. Then we have 
\begin{align*}
   \frac{\Pr[M(\theta)=o_w]}{\Pr[M(\theta')=o_w]} &\leq \frac{\EXP^{\theta'}(L_{\ell})}{\EXP^{\theta}(L_{\ell})} \\
   & \leq \frac{\sum_{j \in L_{\ell} } \exp(\epsilon (\sigma(\theta, o_j)+ \Delta \sigma))}{\sum_{j \in L_{\ell} } \exp(\epsilon \sigma(\theta, o_j))} \\
   & \leq \exp(\epsilon \Delta \sigma)
\end{align*}
\\
(2) $v_i > v_i'$. In this case, the second ratio is at most $1$. Then we have 
\begin{align*}
   {\frac{\Pr[M(\theta)=o_w]}{\Pr[M(\theta')=o_w]}} &\leq \frac{\EXP^{\theta}(w)}{\EXP^{\theta'}(w)} 
   \leq \frac{\exp(\epsilon \sigma(\theta, o_w))}{\exp(\epsilon (\sigma(\theta, o_w) - \Delta \sigma))} \\
   &\leq \exp(\epsilon \Delta \sigma)
\end{align*}
\end{proof}

\section{Proof of Theorem~\ref{thm:sw}}
\label{app:thm:sw}

{\bf Theorem~\ref{thm:sw}} {\em Given a global profile $\theta$, the expected social welfare of layered DPDM $\Mlay$ is at least $\gamma_{d_{\max}}{\E_{\EMD}[sw_{\EMD}(\theta)]}$.}

\begin{proof}
Given a global profile $\theta$, the expected social welfare of $\Mlay$ is
\begin{equation*}
\begin{aligned}
\E_{\Mlay}[sw_{\Mlay}(\theta)] &=\sum_{i\in V}{\left(v_i \times \Pr^{\Mlay}_{i}(\theta_i)\right)} \\
&=\sum_{i\in V}{ v_i \frac{ \exp(\epsilon,\sigma(\theta,o_i))}{\sum_{j\in L_{d_i}}{\frac{1}{\gamma_{d_i}}} \exp(\epsilon,\sigma(\theta,o_j))}}\\
&\geq \gamma_{d_{\max}} \sum_{i\in N}{ v_i \frac{ \exp(\epsilon,\sigma(\theta,o_i))}{\sum_{j\in L_{d_i}}{\exp(\epsilon,\sigma(\theta,o_j))}}} \\
&\geq \gamma_{d_{\max}} \sum_{i\in N}{ v_i \frac{ \exp(\epsilon,\sigma(\theta,o_i))}{\sum_{j\in V}{\exp(\epsilon,\sigma(\theta,o_j))}}}\\
&= \gamma_{d_{\max}} \E_{\Mlay}[sw_{\Mlay}(\theta)]
\end{aligned}
\end{equation*}
\end{proof}

\end{document}